\documentclass[runningheads]{llncs}

\usepackage[T1]{fontenc}
\usepackage{graphicx}
\usepackage{amsmath}
\usepackage{amssymb}
\usepackage{mathtools}
\usepackage{tikz}
\usetikzlibrary{graphs,quotes,positioning}

\newcommand{\N}{\mathbb{N}}
\newcommand{\Z}{\mathbb{Z}}
\renewcommand{\odot}{\odot}
\newcommand{\sh}{\mathsf{sh}}
\newcommand{\dom}{\mathsf{dom}}
\newcommand{\rhs}{\mathsf{rhs}}

\def \<#1>{{\langle {#1} \rangle}}

\newcommand{\comp}{\mathbin{;}}
\newcommand{\arity}{\text{rank}}

\newcommand{\fin}{\text{fin}}

\begin{document}

\title{Shape Preserving Tree Transducers}

\author{Paul Gallot\inst{1}\orcidID{0000-0001-6742-4885}
\and
Sebastian Maneth\inst{1}}

\institute{Department of Mathematics and Informatics, University of Bremen}

\maketitle

\begin{abstract}
It is shown that shape preservation is decidable for
top-down tree transducers, bottom-up tree transducers, and for
compositions of total deterministic macro tree transducers.
Moreover, if a transducer is shape preserving, then it can be
brought into a particular normal form, where every input node
creates exactly one output node.
\end{abstract}

\section{Introduction}

Tree transducers are a well established formalism for describing translations
on finite ordered ranked trees. They were invented in the 1970s in the context of
compilers and mathematical linguistics.
An important type of problem that has been of continuing interest 
are so called ``definability problems'':
given two classes $X$ and $Y$ of translations we want to know whether or not
it is decidable for any $\tau\in X$, whether $\tau$ is definable in
the class $Y$, i.e., whether or not $\tau\in Y$. If such decidability holds
then we say that ``$Y$ is decidable within $X$''.
Knowing that $\tau\in Y$ may be beneficial for several reasons.
For instance elements from $Y$ may be more efficient to implement than
elements from $X$; or, $Y$ may enjoy better closure properties than $X$.

Examples of decidable definability problems on strings are:
the class of regular languages is decidable within the class
of deterministic context-free languages~\cite{DBLP:journals/jacm/Valiant75} and
the class of one-way string transductions is
decidable within the class of functional two-way 
transductions~\cite{DBLP:conf/lics/FiliotGRS13,DBLP:journals/lmcs/BaschenisGMP18}.
For tree transducers:
the class of macro tree translations~\cite{DBLP:journals/jcss/EngelfrietV85} of linear size increase is decidable
within the class of macro tree translations~\cite{DBLP:journals/siamcomp/EngelfrietM03} 
(in fact, even within the composition closure of the latter~\cite{DBLP:journals/acta/EngelfrietIM21});
this result has recently been generalized to linear height increase and to linear
size-to-height increase~\cite{DBLP:conf/icalp/GallotM0P24}.
It was shown that linear deterministic top-down tree translations are decidable
within top-down tree translations~\cite{DBLP:journals/ijfcs/ManethSV23}.

In this paper we study \emph{shape preservation}: a tree translation is shape preserving
if the shape of every output tree coincides with the shape of the corresponding input tree.
Thus, the output and input trees have the same nodes and edges, but may only differ in the labels of
the nodes. It was shown by F{\"u}l{\"o}p and Gazdag~\cite{DBLP:journals/tcs/FulopG03} that the class of shape preserving top-down
tree translations is decidable within the class of top-down tree 
translations~\cite{DBLP:journals/mst/Rounds70,DBLP:journals/jcss/Thatcher70}.
They also show that this class can be realized by ``relabling TOPs'' (see below).
Their proofs go through a series of technically involved constructions and lemmas.

In this paper we give an alternative proof of the decidability of shape preservation
for top-down tree transducers (TOPs). It is based on the facts that
(\emph{i})~functionality of TOPs is decidable and that
(\emph{ii})~equivalence of functional TOPs is decidable. Both were proven originally by 
\'{E}sik~\cite{DBLP:journals/actaC/Esik81}.
The idea is extremely simple: for a given TOP, we change the labels of all output nodes to
a fixed label. The resulting TOP must be functional, if the given one is shape preserving.
Moreover, it must be equivalent to a fixed shape preserving transducer that only depends on the domain
of the given transducer. This ``proof scheme'' can readily be applied to other classes:
e.g., we can show that shape preservation for bottom-up tree transducers
is decidable. In fact, we can even show that shape preservation for (compositions of)
total deterministic macro tree transducers is decidable.

Next, we want to show that every shape preserving TOP can
be realized by a relabeling TOP, i.e., one where every rule is of the form
$q(f(x_1,\dots x_k))\to f'(q_1(x_1),\dots, q_k(x_k))$.
This was already shown by F{\"u}l{\"o}p
and Gazdag, however only for \emph{linear TOPs}.
It can be tempting to assume that every shape preserving TOP 
is linear, but it is not true:
\[
\begin{array}{lcllcl}
q_0(a(x_1)) &\to& a(f(q_1(x_1), q_2(x_1)))\qquad\qquad & 
q_1(f(x_1,x_2)) &\to& q_0(x_1)\\
q_0(e) &\to& e &
q_2(f(x_1,x_2)) &\to& q_0(x_2) 
\end{array}
\]
This TOP is shape preserving (it realizes the identity on trees with paths
of odd length and consisting of alternating $a$- and $f$-nodes and a final $e$-node).
Its first rule is \emph{not} linear: 
the input variable $x_1$ appears more than once in the right-hand side.
The TOP is also \emph{not} nondeleting (not every $x_i$ from a 
left-hand side appears in the right-hand side).
Lemma~3.1 of~\cite{DBLP:journals/tcs/FulopG03} states that every shape preserving TOP is nondeleting.
However, this lemma is true only for linear TOPs and is 
invalidated by the example above for non-linear ones.
To fix this omission, we prove that for every shape preserving TOP, an
equivalent linear TOP can be constructed.
Together with the results of~\cite{DBLP:journals/tcs/FulopG03} this establishes that every
shape preserving TOP is effectively equivalent to a relabeling TOP.

Let us now consider shape preservation for
total deterministic macro tree transducers (MTTs).
The property is decidable, because linear size increase of (compositions of) MTTs
is decidable, and in the affirmative case we are in a class
of translations that has decidable equivalence~\cite{DBLP:journals/ipl/EngelfrietM06}.
A question that remains: is there a normal form for MTTs that is
akin to relabeling TOPs?
Consider the following example of an MTT.
It takes as input trees of the form $\$(t)$ where $t$ is a monadic
tree over unary labels $a$, $b$, and $c$ (and the nullary label $e$).
\[
\begin{array}{lcll}
q_0(\$(x_1)) &\to& \$(q_a(x_1, q_b(x_1, q_c(x_1, e))))\\
q_\delta(\gamma(x_1),y) &\to& \delta(q_\delta(x_1, y)) &\text{for }\delta,\gamma\in\{a,b,c\}\text{ with }\delta=\gamma\\
q_\delta(\gamma(x_1),y) &\to& q_\delta(x_1, y) &\text{for }\delta,\gamma\in\{a,b,c\}\text{ with }\delta\not=\gamma\\
q(e,y) &\to& y &\text{for }q\in\{q_0,q_a,q_b,q_c\}
\end{array}
\]
For instance, the input tree $\$(c(a(b(a(a(b(e)))))))$ is translated to the
output tree $\$(a(a(a(b(b(c(e)))))))$.
Our corresponding normal form of ``relabeling MTT'' requires that each input node
produces exactly one output node; but this output node may be anywhere in the rule,
including a ``parameter position''.

\section{Preliminaries}

The set $\{0,1,\dots \}$ of natural numbers is denoted by $\N$.
For $k\in\N$ we denote by $[k]$ the set
$\{1,\dots,k\}$; thus $[0]=\emptyset$.
We fix the set $X=\{x_1,x_2,\dots \}$ of variables
and the set $Y=\{y_1,y_2,\dots\}$ of parameters and assume these
sets to be disjoint from all other alphabets and sets.
For $k\geq 1$ let $X_k=\{x_1,\dots,x_k\}$ (and similarly for $Y$).
A ranked alphabet (set) consists of an alphabet (set) $\Sigma$ together
with a mapping $\text{rank}_{\Sigma}: \Sigma\to\N$
that assigns to each symbol $\sigma\in\Sigma$ a natural number called its ``rank''.
By $\Sigma^{(k)}$ we denote the symbols of $\Sigma$ that have rank $k$.

The set $T_\Sigma$ of (finite, ranked, ordered) trees over $\Sigma$ is the smallest
set of strings $S$ such that if $\sigma\in\Sigma^{(k)}$, $k\geq 0$, and
$s_1,\dots,s_k\in S$, then also $\sigma(s_1,\dots,s_k)\in S$.
We write $\sigma$ instead of $\sigma()$.
Let $A$ be a set that is disjoint from $\Sigma$. Then the set $T_\Sigma(A)$ of
trees over $\Sigma$ indexed by $A$ is defined as $T_{\Sigma'}$ where
$\Sigma'=\Sigma\cup A$ and $\text{rank}_{\Sigma'}(a)=0$ for $a\in A$ and 
$\text{rank}_{\Sigma'}(\sigma)=\text{rank}_{\Sigma}(\sigma)$ for $\sigma\in\Sigma$.
For a tree $s=\sigma(s_1,\dots,s_k)$ with $\sigma\in\Sigma^{(k)}$, $k\geq 0$,
and $s_1,\dots,s_k\in T_\Sigma$, we define 
the set $V(s)\subseteq\N^*$ of nodes of $s$ as
$\{\varepsilon\}\cup \{ iu\mid i\in[k], u\in V(s_i)\}$;
thus, nodes are strings over positive integers,
where $\varepsilon$ denotes the root node of $s$, and
for a node $u$, $ui$ denotes the $i$-th child of $u$.
For $u\in V(s)$ we denote by $s[u]$ the label of $u$ in $s$ and
by $s/u$ the subtree rooted at $u$.
We denote the prefix order on paths by $\leq$. 
Formally, let $s=\sigma(s_1,\dots,s_k)$ and define 
$s[\epsilon]=\sigma$,
$s[iu]=s_i[u]$, 
$s/\varepsilon=s$, and
$s/iu=s_i/u$ for $\sigma\in\Sigma^{(k)}$, $k\geq 0$,
$s_1,\dots,s_k\in T_\Sigma$, $i\geq 1$ and $u\in V(s_i)$ such that $iu\in V(s)$.
For trees $s,t$ and $u\in V(s)$ we denote by
$s[u\leftarrow t]$ the tree obtained from $s$ by replacing its subtree
rooted at $u$ by the tree $t$.
Similarly, for trees $t_1,\dots t_k$ we denote by
$s[x_i\leftarrow t_i]_{i\leq k}$ the tree obtained from $s$ by replacing
each occurrence of $x_i$ by $t_i$.

For each $n \in \N$ we fix the shape alphabet $\Psi_n = \{\#_0, \#_1, \dots, \#_n\}$ with $\arity(\#_i) = i$ for all $i \leq n$ and, for each alphabet $\Sigma$ of maximal arity $n$, the shape function $\sh_\Sigma : T_\Sigma \to T_{\Psi_n}$ on $\Sigma$ inductively defined by, for all $k \geq 0$, $\sigma \in \Sigma^{(k)}$ and trees $t_1, \dots, t_k \in T_\Sigma$, $\sh_\Sigma(\sigma(t_1, \dots, t_k)) = \#_k(\sh_\Sigma(t_1), \dots, \sh_\Sigma(t_k))$. 
We note $SH$ the set of all shape functions. 

The domain of a relation $f$ is noted $\dom(f)$.
The composition of relations is noted $\comp$ and is defined so that $(x,z) \in f \comp g$ when $\exists y \in \dom(g), (x,y) \in f$ and $(y,z) \in g$. 
A tree relation $f \subseteq T_\Sigma \times T_\Delta$ over alphabets $\Sigma$ and $\Delta$ is \emph{shape preserving} if $f \comp \sh_\Delta \subseteq \sh_\Sigma$.
\textbf{Convention:} All inclusions of transducer classes mentioned in this paper are meant effective.

A \emph{top-down tree transducer} is a system $M = (Q, \Sigma, \Delta, q_0, R)$, where
$Q$ is a finite set of \emph{states},
$\Sigma$ and $\Delta$ are the \emph{input and output ranked alphabets}, 
$q_0\in Q$ is the \emph{initial state}, and $R$ is a finite set of \emph{rules} of the form:
$q(\sigma(x_1,\dots, x_k))\to\rhs$,
$\sigma\in\Sigma^{(k)}$, $k\geq 0$,
$q\in Q$, and $\rhs \in T_\Delta(\<Q,X_k>)$, with $\<Q,X_k> = \{ q'(x) \mid q' \in Q, x \in X_k\}$. 
For every state $q \in Q$ and tree $t \in T_\Sigma(X)$,
$q(t)$ is called a \emph{state call} of $q$ on $t$.
The set of state calls of $M$ is written $\<Q,T_\Sigma(X)>$.
%
For a rule $r$ of the form $q(\sigma(x_1 , \dots, x_k)) \to \rhs$ in $R$
and tree $t \in T_\Delta(\<Q,T_\Sigma(X)>)$ such that $t/u = q(\sigma(t_1 , \dots, t_k))$
for some $u\in V(t)$ and trees $t_1, \dots, t_k \in T_\Sigma(X)$, we write $t \xrightarrow{r,u} t[u \leftarrow \rhs[x_i \leftarrow t_i]_{i \leq k}]$.

A run of a state $q\in Q$ on a tree $t\in T_\Sigma(X)$ is of the form
$q(t)\xrightarrow{r_1,u_1} t_1\dots \xrightarrow{r_n,u_n} t_n$.
We also write $t\xrightarrow{q} t_n$ or $t\to^* t$ and call $t_n$
the \emph{output} of the run. If $t_n\in T_\Delta$ then the run
is called \emph{valid}.
Let $M_q = \{ (t,t') \mid$ there is a valid run of $q\text{ on }t\text{ with output } t' \}$. 
The \emph{translation realized by $M$} is $M_{q_0}$. 
%
%
For each valid run $r$ on $t$ with output $t'$ and for each $u \in V(t)$,
there exists two runs $t[u \leftarrow x_1] \to^* t_1$ and
$t_1[x_1 \leftarrow t/u] \to^* t'$ with $t_1 \in T_\Delta(\<Q,\{x_1\}>)$
The first such run is called \emph{provisional run}
and $t_1$ the \emph{provisional output} of $\xi$ at path $u$.
The sequence of states appearing in $t_1$ (in pre-order)
is called \emph{state sequence} of $r$ at $u$. 
Provisional runs only contain one variable $x_1$, we usually write it $x$ to avoid
confusion with variables $x_1, \dots, x_k$ used in the rules of $M$. 

\begin{example}
Consider the TOP from the Introduction and number its rules
from left to right and then from top to bottom.
Then
\begin{multline*}
q_0(a(f(e,e)))\xrightarrow{r_1,\varepsilon}
a(f(q_1(f(e,e)),q_2(f(e,e))))\xrightarrow{r_2,11}
a(f(q_0(e),q_2(f(e,e))))\\
\xrightarrow{r_4,12}
a(f(q_0(e),q_0(e)))\xrightarrow{r_3,1}
a(f(e,q_0(e)))\xrightarrow{r_3,2}
a(f(e,e)).
\end{multline*}
\end{example}

Let $\text{ID}$ denote the class of all total identity functions on $T_\Sigma$, for
any ranked alphabet $\Sigma$

\begin{lemma}\label{lem:id}
All mappings in $\text{ID}$ and in \text{SH} can be realized
by total deterministic nondeleting and linear top-down tree transducers.
\end{lemma}
\begin{proof}
Let $\Sigma$ be a ranked alphabet and let $n$ be the
maximal rank of symbols in $\Sigma$.
For $i\in\{1,2\}$ define
$M_i=(\{q\},\Sigma,\Delta_i,q,R_i)$ where
$\Delta_1=\Sigma$ and $\Delta_2=\Psi_n$.
For every $\sigma\in\Sigma^{(k)}$ with $k\geq 0$ let
$q(\sigma(x_1,\dots,x_k))\to\sigma(q(x_1),\dots,q(x_k))$
be in $R_1$ and let
$q(\sigma(x_1,\dots,x_k))\to\#_k(q(x_1),\dots,q(x_k))$ be in $R_2$.
Then $M_1$ realizes the total identity on $T_\Sigma$ and
$M_2$ realizes $\sh_\Sigma$.
\qed
\end{proof}

\section{Decidability of Shape Preservation}

Let ${\cal C}$ be a class of relations on trees, i.e., for any
$c\in {\cal C}$, $c\subseteq T_\Sigma\times T_\Delta$ for some
ranked alphabets $\Sigma$ and $\Delta$.
For $c\in{\cal C}$ let $\text{dom}(c)$ denote the domain of $c$ and
let $\text{DOM}(c)$ denote this domain, seen as a partial identity function,
i.e., $\text{DOM}(c)=\{(s,s)\mid s\in\text{dom}(c)\}$.
We denote by $\text{DOM}({\cal C})$ the class of domains of ${\cal C}$,
seen as partial identity functions,
$\text{DOM}({\cal C}) = \bigcup_{c\in{\cal C}} \text{DOM}(c)$.

\begin{theorem}\label{the:decide_sp}
Let ${\cal C}$ be a class of relations on trees.
Shape preservation is decidable for every element of ${\cal C}$ if:  
\begin{enumerate}
\item $\text{ID}\subseteq {\cal C}$,
\item ${\cal C} \comp \text{SH} \subseteq {\cal C}$, 
\item $\text{DOM}({\cal C}) \comp {\cal C} \subseteq {\cal C}$,
\item functionality is decidable for any $c\in{\cal C}$, and
\item equivalence is decidable for all functions in $C$.
\end{enumerate}
\end{theorem}
\begin{proof}
Let $c\in{\cal C}$ and let $\Sigma,\Delta$ be ranked alphabets
(of minimal size) such that $c\subseteq T_\Sigma\times T_\Delta$.
Define $d = c \comp \sh_\Delta$.
By Property~(2), $d\in{\cal C}$.
By Property~(4) it is decidable whether or not $d$ is a function.
If $d$ is not a function, then $c$ is not shape preserving.
If $d$ is a function, then let $e=\text{DOM}(c) \comp \sh_{\Sigma}$.
It follows from Properties~(1) and (2) that $\sh_{\Sigma}\in{\cal C}$
and hence $e\in{\cal C}$ by Property~(3).
Clearly, $c$ is shape preserving if and only if
$d$ is equivalent to $e$; the latter is decidable by Property~(5).
\qed
\end{proof}

Let $\text{FTA}$ denote all partial identity functions on regular tree languages.
Let us now consider the class ${\cal C}=\text{FTA} \comp \text{TOP}$.
By Lemma~\ref{lem:id}, $\text{ID}\subseteq \text{TOP}$ and
therefore $\text{ID}\subseteq{\cal C}$.
Hence, Property~(1) is satisfied.
Since every element of $\text{SH}$ can be realized by a linear and
nondeleting top-down tree transducer,
it follows from Theorem~1 of~\cite{DBLP:journals/iandc/Baker79b} that Property~(2) is
satisfied.
Property~(3) follows from the fact that FTA is closed under composition
(Lemma~3.4 of~\cite{DBLP:journals/mst/Engelfriet75}).
Property~(4) follows from the proof of Theorem~8 of~\cite{DBLP:journals/actaC/Esik81},
where it is shown that functionality of ${\cal C}$ is decidable.
By the (unique) Theorem of~\cite{DBLP:journals/ipl/Engelfriet78} the restriction of $\text{TOP}$
to functions is included in the class of deterministic top-down
tree transducers with regular look-ahead. Equivalence of the latter
class is decidable (see~\cite{DBLP:journals/jcss/EngelfrietMS09}), so Property~5 is satisfied.

\begin{corollary}\label{cor:TOP}
Shape preservation is decidable for top-down tree translations (restricted
to arbitary regular input tree languages).
\end{corollary}

Note that for two tree relations $r_1,r_2$ such that $r_1$ is shape preserving
it holds that $r_1\comp r_2$ is shape preserving if and only if $r_2$ is 
shape preserving on the range of $r_1$.
Top-down tree transducers \emph{with regular look-ahead}
are included in the composition of deterministic bottom-up finite state relabelings and 
$\text{TOP}$ by Theorem~2.6 of~\cite{DBLP:journals/mst/Engelfriet77}.
Such relabelings are shape preserving and their ranges are effectively regular.
Thus, decidability of shape preservation can be reduced from top-down tree transducers with
regular look-ahead to ${\cal C}=\text{FTA} \comp \text{TOP}$, which is solved
by Corollary~\ref{cor:TOP}.
Similarly, by Theorem~3.15 of~\cite{DBLP:journals/mst/Engelfriet75}, \emph{bottom-up tree translations} are included
in compositions of bottom-up finite state relabelings and tree homomorphisms; since tree homomorphisms are
particular TOPs, decidability of shape preservation again reduces to that of ${\cal C}$.

\begin{corollary}
Shape preservation is decidable for (1)~top-down tree translations with regular
look-ahead and for (2)~bottom-up tree translations. 
\end{corollary}

Let $\text{MTT}$ denote the class of total deterministic macro tree transducers (full definition in section~\ref{sec:mtt_normalization})
and let $\text{LMTT}$ denote the restriction of $\text{MTT}$ to
linear size increase (LSI).
Since every shape preserving translation is of LSI,
we first check the LSI property.
This is possible even for compositions of MTTs and 
in the affirmative case we obtain one equivalent LMTT,
according to Theorems~44 and~43 of~\cite{DBLP:journals/acta/EngelfrietIM21}.
That Theorem~44 states that one ``dTTsu'' transducer can be constructed.
This is a ``single-use tree-walking transducer (with regular look-around)'',
which is the same as ``single use restricted attributed tree transducer with look-ahead'';
according to Theorem~7.1 of~\cite{DBLP:journals/iandc/EngelfrietM99} 
they are the same as (deterministic) MSO definable tree transductions
which are the same as LMTTs by~Theorem~7.1 of~\cite{DBLP:journals/siamcomp/EngelfrietM03}.

We consider the class ${\cal C} = \text{FTA} \comp \text{LMTT}$.
These are exactly the partial deterministic MSO tree translations, because
$\text{LMTT}$ coincides with the total deterministic MSO tree translations (see above)
and the domains of partial MSO tree translations are exactly the regular
tree languages (because an explicit
closed ``MSO domain formula'' defines the domain and MSO definable tree languages
are exactly the regular tree languages~\cite{DBLP:journals/jcss/Doner70,DBLP:journals/mst/ThatcherW68}).
Since every total deterministic TOP is also an MTT, and mappings in ID are LSI,
it follows from Lemma~\ref{lem:id} that Property~(1) of Theorem~\ref{the:decide_sp} holds.
Property~(2) follows from the fact that (deterministic) MSO definable tree translations
are closed under composition by Proposition~3.2 of~\cite{DBLP:journals/tcs/Courcelle94}.
Property~(3) holds because $\text{FTA}$ is closed under composition (see above).
Property~(4) is void, because ${\cal C}$ contains functions only.
Equivalence of $\cal C$ is decidable~\cite{DBLP:journals/ipl/EngelfrietM06} thus
giving us Property~(5).

\begin{corollary}
Shape preservation is decidable for compositions of total deterministic macro
tree transducers.
\end{corollary}

\newcommand{\A}{\mathcal{A}}

\section{Normalization of \emph{non-linear} shape preserving TOPs}

The paper~\cite{DBLP:journals/tcs/FulopG03} provides a normalization algorithm for linear TOPs, i.e.\ showing how to transform a shape preserving linear TOP into an equivalent nondeterministic top-down relabeling. 
We only provide here a way to transform a non-linear shape preserving TOP into an equivalent linear shape preserving TOP. 
Our construction relies on determining which states of a TOP can only produce a bounded number of outputs, computing these outputs using a form of look-ahead automaton, and proving that other state calls on the same input subtree can be merged into one single state call. 

Throughout this section we will use as running example the non-linear shape preserving TOP $M_0$ with set of states $Q_0 = \{q_0,q_1,q_2,q_3,q_\ell\}$, both input and output alphabet $\Sigma = \{f,g,h,a,b\}$, and the following rules:
\begin{align*}
q_0(h(x)) &\to q_1(x) ~~~&~~~ q_1(h(x)) &\to h(h(f(q_2(x), q_3(x)))) \\
q_2(f(x_1,x_2)) &\to f(q_2(x_1), q_\ell(x_2)) ~~~&~~~ q_3(f(x_1, x_2)) &\to q_3(x_1) \\
q_2(g(x_1, x_2)) &\to q_\ell(x_2) ~~~&~~~ q_3(g(x_1, x_2)) &\to q_\ell(x_1) \\
q_\ell(a) &\to a ~~~&~~~ q_\ell(b) &\to b
\end{align*}

We assume given a non-linear shape preserving TOP $M$ on tree alphabets $(\Sigma,\Delta)$, and we assume all states and rules which do not appear in a valid run of $M$ have been removed. 

As hinted at in our TOP example $M_0$, there is a one-to-one correspondence between leaves in the input and leaves in the output:

\begin{lemma}\label{lem:leafs}
If $t \xrightarrow{M} t'$ with $t\in T_\Sigma$ and $t' \in T_\Delta$ then each input leaf is processed by exactly one state call, which itself produces an output tree containing exactly one output leaf. 
\end{lemma}
\begin{proof}
If a leaf is not processed by any state of $M$, then it can be replaced with a larger tree without changing the output, which contradicts the shape preserving property. 
Each state call to a leaf must produce at least one leaf in the output, so each input leaf produces at least one output leaf. 
The shape preserving implies that the number of leaves in the input is the same as in the output, so each input leaf produces exactly one output leaf. \qed
\end{proof}

Lemma~\ref{lem:leafs} has a number of interesting implications.
For instance, it implies that any rule on a non-leaf input node cannot
produce any leaf; if it did, then the lemma implies that the output tree
has more leaves than the input tree -- a contradiction.
It also implies that no state can be called twice on the same input node in a valid run.
Indeed, if the same state appeared twice in a state sequence, then both state calls could have
the same run on their input. In that run, at least one output leaf is produced when
processing an input leaf $u$. This would imply that the input leaf $u$ is processed twice,
thus contradicting the lemma.

\begin{lemma}\label{lem:state_sets}
For each state sequence $q_1, \dots, q_n$ of $M$ and $1 \leq i < j \leq n$: $q_i \neq q_j$. 
\end{lemma}


This allows us to talk about \emph{state sets} of $M$ instead of state sequences of $M$. 
We denote by $S_Q$ the set of \emph{state sets} of $M$, i.e.\ $S_Q = \{ \{q_1, \dots, q_n\} \subseteq Q \mid q_1, \dots, q_n \text{ is a state sequence of } M \}$. Note that the set $S_Q$ is computable from $M$. 
For each $S \in S_Q$ we define $I(S) = \bigcap_{q \in S} \dom(M_q)$ and $\fin(S)$ as the set of states $q \in S$ such that $M_q(I(S))$ is a finite set of trees. For each $q \in Q$ the domain $\dom(M_q)$ is regular and so is $I(S)$. For each $S \in S_Q$, by deciding the finiteness of ranges of TOPs~\cite{DBLP:journals/iandc/DrewesE98}
we can compute $\fin(S)$ and $M_q(I(S))$ for all $q \in \fin(S)$. 
Finally, for every $S \in S_Q$, every $q \in \fin(S)$ and every $t \in M_q(I(S))$, the tree language $M_q^{-1}(\{t\})$ is regular and we can compute a nondeterministic top-down tree automaton for this language. 

We compute this on example $M_0$: $S_Q = \{ \{q_0\}, \{q_1\}, \{q_2,q_3\}, \{q_\ell\} \}$, 
with $\fin(\{q_0\}) = \emptyset$, $\fin(\{q_1\}) = \emptyset$, $\fin(\{q_2,q_3\}) = \{q_3\}$ and $\fin(\{q_\ell\}) = \{q_\ell\}$. We have $M_{0q_3}(I(\{q_2,q_3\})) = \{a,b\}$ and $M_{0q_\ell}(I(\{q_\ell\})) = \{a,b\}$. 

For every provisional run $t[u \leftarrow x] \to t'[u_i \leftarrow q_i(x)]_{i \leq n}$ and for every $i \leq n$,
we say that $u$ is a \emph{origin} of the state call $q_i(x)$ at path $u_i$,
and we say that this state call has \emph{shape distance} $n$, if $n$ is
the distance in the tree $t$ between the origin $u$ and the node at path $u_i$. 

We can then show that the shape distance is bounded except for state calls of states in $\fin(S)$. We can see that the exception of states in $\fin(S)$ is necessary on the example $M_0$: 

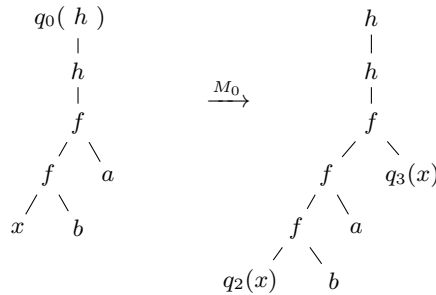
\vspace{-4mm}
\begin{center}
\begin{tikzpicture}
	[level distance=7mm,sibling distance=8mm]
	\small
	\node at (0,0) {$q_0(~ h ~)~~$}
	child {node {$h$}
		child {node {$f$}
			child {node {$f$}
				child {node {$x$}}
				child {node {$b$}}
			}
			child {node {$a$}}
		}
	};
	
	\node at (2,-1
	) {$\xrightarrow{M_0}$};
	
	\node at (3.9,0) {$h$}
	child {node {$h$}
		child {node {$f$}
			child[sibling distance=12mm] {node {$f$}
				child[sibling distance=8mm] {node {$f$}
					child[sibling distance=10mm] {node {$q_2(x)~~$}}
					child[sibling distance=10mm] {node {$b$}}
				}
				child[sibling distance=8mm] {node {$a$}}
			}
			child[sibling distance=12mm] {node {$q_3(x)~$}}
		}
	};
\end{tikzpicture}
\end{center}

\vspace{-4mm}
Here we can see that the state call $q_3(x)$ is far from its origin $u = 1111$, and indeed we can make state calls $q_3(x)$ arbitrarily far from each other by lengthening the string of $f$ nodes in the input. On the other hand, $q_2(x)$ is always one step below its origin. 

Formally, we prove a bound on the shape distance for state calls of states not in $\fin(S)$, and we show that the path $u'$ to such a state call in the output is either a prefix of its origin $u$, or $u$ is a prefix of $u'$. 

\begin{lemma}\label{lem:origin_distance}
For every input tree $t$, run $r$ of $M$ on $t$ and its provisional output $t[u \leftarrow x] \xrightarrow{M} t'[u_i \leftarrow q_i(x)]_{i \leq n}$ at a path $u \in V(t)$, noting $S$ the state set of run $r$ at path $u$, for every $i$ such that $q_i \notin \fin(S)$ we have $u$ and $u_i$ are comparable wrt.\ the prefix ordering and the corresponding shape distance is $\leq\max(\rhs)\cdot 2^{|Q|}$
where $\max(\rhs)$ is the maximum height of the right-hand side of a rule of $M$. 
\end{lemma}

With this lemma we can merge state calls on copies of input subtrees:

\begin{theorem}\label{the:copy-removal}
For every non-linear shape preserving TOP $M$, we can compute a shape preserving linear TOP $M'$ equivalent to $M$. 
\end{theorem}

The formal construction and its full proof are in Appendix. We show here the construction on an example and give intuition for it. 
The idea of this construction is to use nondeterminism to guess the outputs of state calls whose set of possible outputs is finite (states $q \in \fin(S)$ where $S$ is the state set of the current input node). The reverse images of possible outputs of such states are regular, so we can check them using a nondeterministic top-down automaton, which we can simulate with our TOP since it has to visit every input node (according to Lemma~\ref{lem:leafs}). 
We then use Lemma~\ref{lem:origin_distance} to merge remaining state calls on copied input subtrees. Indeed Lemma~\ref{lem:origin_distance} implies that for each input subtree $t$, there is a subtree $t'$ in the output which contains all state calls $q(t)$ on input $t$ with $q \in S \setminus \fin(S)$, and there is a bound on the depth of $q(t)$ in $t'$. So all such state calls can be merged into one state call which outputs $t'$. 

\paragraph{The construction on example TOP $M_0$:}
Note that the output of state $q_3$ is always either $a$ or $b$. So, for state set $S = \{q_2, q_3\}$, $\fin(S) = \{q_3\}$ and we can check the output of $q_3$ by checking a regular property on the input. Similarly, for $S' = \{q_\ell\}$, we have $\fin(S_\ell) = \{q_\ell\}$. 

We now build the linear TOP $M'$ equivalent to $M_0$. States of $M'$ are of the form $(\phi, v, t)$, with the state call $(\phi, v, t)(t_0)$ meaning: 
\begin{itemize}
\item $\phi$ is a regular property checked on $t_0$, it allows us to produce the output of state calls of $M_0$ on $t_0$ whose output is bounded (states in $\fin(S)$), 
\item $t$ 
is an output subtree containing all state calls of $M_0$ on $t_0$, except for states in $\fin(S)$, 
\item $v$ is a path such that, if $t_1[u \leftarrow x] \xrightarrow{M'} t'_1[u' \leftarrow (\phi, v, t)(x)]$, then $u = u'\,v$. 
\end{itemize}
To put it simply, for a state call at path $u'$ with origin at path $u$, if $u$ is a prefix of $u'$ then $v = \varepsilon$ and $t$ gives the subtree at path $u$ in the provisional output; but if $u'$ is a prefix of $u$ then $u = u'\,v$ and $t=q(x)$ where $q(x)$ is the only state call with origin $u$. 
Here are some of the rules for $M'$:
\[
\begin{array}{lcl}
(\emptyset, \varepsilon, q_0(x))(h(x_1)) &\to& (\emptyset,1,q_1(x))(x_1)\\
(\emptyset, 1, q_1(x))(h(x_1)) &\to&h(h((q_3(x)=a,\varepsilon,f(q_2(x),a))(x_1)\\
(\emptyset, 1, q_1(x))(h(x_1)) &\to& h(h((q_3(x)=b,\varepsilon,f(q_2(x),b))(x_1)\\
(q_3(x)=a,\varepsilon,f(q_2(x),a))(f(x_1,x_2)) &\to& f((q_3(x)=a,\varepsilon,f(q_2(x),b))(x_1),\\
&&(q_l(x)=b,\varepsilon,a)(x_2))
\end{array}
\]



\begin{theorem}
For every shape preserving TOP $M$, we can construct
an equivalent (nondeterministic) top-down relabeling equivalent to $M$. 
\end{theorem}
\begin{proof}
We can construct a linear TOP $M'$ from $M$ by applying Theorem~\ref{the:copy-removal}. 
According to Theorem~3.30 from~\cite{DBLP:journals/tcs/FulopG03},
which works on linear shape preserving TOPs,
there exists a top-down relabeling equivalent to $M'$.
\qed
\end{proof}

The construction in~\cite{DBLP:journals/tcs/FulopG03} uses properties of shape preserving TOPs similar to the ones we use here. 
The two constructions are similar enough that
we believe they could be merged into a single-step transformation transforming
an arbitrary shape preserving TOP directly into a relabeling TOP.

\section{Normalization of Shape Preserving MTTs}\label{sec:mtt_normalization}

In this section we show how to normalize shape preserving total deterministic MTTs.
As opposed to TOPs, one cannot transform any shape preserving MTT into an equivalent relabeling.
In fact, as can be seen from the example in the Introduction, we cannot even associate
a linear MTTs with every shape preserving one, i.e., copying of input variables is
an essential power of a shape preserving MTT that \emph{cannot} be avoided!
Instead, we find a class of MTTs that we call \emph{one-to-one} which characterizes
shape preserving MTTs: for a one-to-one MTT, each input node creates exactly one output node.
A total deterministic \emph{macro tree transducer} (MTT) 
$M$ is a tuple $(Q,P,\Sigma,\Delta,q_0,R)$, where
$Q$ is a ranked alphabet of \emph{states},
$\Sigma$ and $\Delta$ are ranked alphabet of \emph{input} and \emph{output symbols},
$q_0\in Q^{(0)}$ is the \emph{initial state}, and
 $R$ is the \emph{set of rules}, where for each $q\in Q^{(m)}$, $m\geq 0$,
	$\sigma\in\Sigma^{(k)}$ and $k\geq 0$ there is exactly one rule of the form 
$	\<q,\sigma(x_1,\dots,x_k)>(y_1,\dots,y_m) \to t$
	with $t\in T_{\Delta\cup \< Q,X_k>}(Y_m)$.
	The right-hand side $t$ of such a rule is denoted by
	$\text{rhs}_M(q,\sigma)$
We use a notation that is slightly different from the one used in the Introduction:
instead of, e.g., $q_2(x_2, q_{\text{id}}(x_1))$ we write
$\< q_2,x_2>(\< q_{\text{id}},x_1>)$.
Thus, we use angular brackets $\< \dots >$ to indicate a state call on an input subtree,
and use round brackets (after the angular brackets), to indicate the parameter arguments
of the particular state call.
Formally, if $\Gamma$ is a ranked alphabet and $A$ a set, then
$\langle \Gamma, A\rangle$ is a ranked alphabet of symbols
$\langle \gamma,a\rangle$ of rank $k$ for all $\gamma\in\Gamma^{(k)}$, $k\geq 0$,
and $a\in A$.

The semantics of a MTT $M$ (as above) is defined as follows.
We define the derivation relation $\Rightarrow_M$ as follows.
For two trees $\xi_1,\xi_2\in T_{\Delta\cup\< Q,T_\Sigma>}(Y)$,
$\xi_1\Rightarrow_M\xi_2$ if there exists a node $u$ in $\xi_1$
with $\xi_1/u=\<q,s>(t_1,\dots,t_m)$, $q\in Q^{(m)}$, $m\geq 0$,
$s=\sigma(s_1,\dots,s_k)$, $\sigma\in\Sigma^{(k)}$, $k\geq 0$,
$s_1,\dots,s_k\in T_\Sigma$,
$t_1,\dots,t_m\in T_{\Delta\cup\<Q,T_\Sigma>}(Y)$, and 
$\xi_2=\xi_1[u\leftarrow\xi]$ where $\xi$ equals
\vspace{-3mm}

\[
\zeta
[\![\<q',x_i>\leftarrow \<q',s_i>\mid q'\in Q,i\in[k]]\!]
[y_j\leftarrow t_j\mid j\in[m]]
\]
and $\zeta=\text{rhs}_M(q,\sigma)$.
The double angular brackets in the displayed formular
merely means that each occurrence of $x_i$ is replaced by $s_i$.
Since $M$ is total deterministic (i.e.,
for every state $q$, input symbol $\sigma\in\Sigma^{(k)}$ and $k\geq 0$, $M$ contains exactly one corresponding rule)
there is for every $\xi_1$ a unique
tree $\xi'\in T_\Delta(Y)$ such that $\xi_1\Rightarrow_M^* \xi'$.
For every $q\in Q^{(m)}$, $m\geq 0$ and $s\in T_\Sigma$ we define the
\emph{$q$-translation of $s$}, denoted by $M_q(s)$, as the unique tree $t$ in $T_\Delta(Y_m)$ such that
$\<q,s>(y_1,\dots, y_m)\Rightarrow_M^* t$.
We denote the \emph{translation realized by $M$} also by $M$, i.e., 
$M=M_{q_0}$ and for every $s\in T_\Sigma$,
$M(s)=M_{q_0}(s)$ is the unique tree $t\in T_\Delta$
such that $\< q_0,s>\Rightarrow_M^* t$. 

Let $s\in T_\Sigma$ and $u\in V(s)$.
Consider the tree $M(s[u\leftarrow x])$. Since there are no rules for the
input symbol $x$, this tree contains state calls in the form
of occurrences of symbols of the form $\langle Q,\{x\}\rangle$.
The state sequence of $M$ on $s$ at path $u$ is the sequence of states $q_1,q_2,\dots$
such that $\langle q_1,x\rangle,\langle q_2,x\rangle\dots$ are all elements
of $\langle Q,\{x\}\rangle$ that appear in $M(s[u\leftarrow x])$ in pre-order.
We define the \emph{origin function} 
$f: V(M(s)) \to V(s)$ associating with each output node the input node which produced it
via a rule application; see~\cite{DBLP:journals/ipl/ManethS23}
for more details and a precise definition.
Intuitively, if we consider the output nodes that are inserted via rule applications
into the tree $M(s[u\leftarrow x])$ when we build
$M(s[u\leftarrow \sigma(x_1,\dots,x_k)])$ for some $\sigma\in\Sigma^{(k)}$.

We illustrate this on an example MTT $M_1$ with the following rules: 
\[\begin{array} {l l}
\<q_0,a(x)> \to \<q_a,x>(\<q_b,x>) & \<q_a,a(x)>(y) \to \<q_a,x>(a(y)) \\
\<q_0,b(x)> \to \<q_m,x>(b(e)) & \<q_a,b(x)>(y) \to \<q_a,x>(y) \\
\<q_0,e> \to e & \<q_a,e>(y) \to y \vspace{1mm}\\
\<q_m,a(x)>(y) \to \<q_m,x>(a(y)) ~~~~~ & \<q_b,a(x)> \to \<q_b,x> \\
\<q_m,b(x)>(y) \to \<q_m,x>(b(y)) & \<q_b,b(x)> \to b(\<q_b,x>) \\
\<q_m,e>(y) \to y & \<q_b,e> \to a(e) \\
\end{array}\]

The MTT $M_1$ translates monadic input trees of the form
$b(w(e))$ into trees $w^r(b(e))$ where $w$ is an arbitary sequence of $a$- and $b$-nodes
and $w^r$ denotes the reverse of the sequence.
On input trees of the form $a(w(e))$ it outputs $a^mb^n(a(e))$ where
$m$ is the number of $a$-nodes in $s$ and $n$ is the number of $b$-nodes in $w$.

\subsubsection{One-to-one MTT}
An MTT $M$ is \emph{one-to-one} if, for every input tree $t\in T_\Sigma$,
corresponding output tree $t'$ and origin function $f:V(t') \to V(t)$: the origin function $f$ is a bijection. 

We can transform any MTT into a composition of a nondeterministic top-down relabeling
with a one-to-one MTT.
The construction relies on adjusting the rate of production of the output.
We define a notion of \emph{delay} in output production.
This \emph{delay} is bounded and we can track it using a top-down relabeling.
We can then use a top-down relabeling to compute pieces of output to produce earlier in order
to adjust the output production rate of the transducer. 

For the remainder of this section, we assume given a shape preserving MTT $M$ over the alphabets $(\Sigma,\Delta)$. 
We extend the definition of the size of trees to trees in $T_\Sigma(X \cup Y)$ such that $|t|$ is the number of nodes in $t$ with labels in $\Sigma$. With this definition, substitution of variables or parameters preserves size. 

We define the \emph{delay} of $M$ on $t \in T_\Sigma$ at path $u \in V(t)$ as the difference between the size of output produced from $t/u$ and the size of $t/u$. We write it $\Delta(t,u)$, so
$\Delta(t,u) = (\sum_{i=1}^{n} |q_i(t/u)|) - |t/u|$ where $q_1, \dots, q_n$ is the state sequence at path $u$. 


On the example $M_1$, the state sequences are $s_1=q_0$, $s_2=q_m$ and $s_3=q_a,q_b$.
Notice that each state sequence always has the same delay: $s_1$
has delay $0$, $s_2$ has delay $-1$ and $s_3$ has delay $1$. 
We can prove this in general. The delay is determined by the state sequence and the input subtree $t/u$, so if two subtrees $t/u$ and $t'/u'$ of two different input trees have the same state sequence but two different delays, then we could substitute $t/u$ for $t'/u'$ and get an output tree of a different size compared to the size of the input tree. 
So the state sequence always determines the delay, and since state sequences are computable by a relabeling, we can also compute the delay using a relabeling.

\begin{lemma}\label{lem:delay_bound}
One can construct a nondeterministic top-down relabeling $R$ associating, with each node at path $u \in V(t)$ in an input tree $t \in T_\Sigma$, the state sequence $q_1, \dots , q_n$ and the delay $\Delta(t,u)$ at that node. 
\end{lemma}

For our construction, for state sequences with positive delay, we will need to produce output earlier than when it is produced by $M$. To do this we change the relabeling so that it predicts parts of the output of states, which we can then output earlier compared to $M$.

On the example MTT $M_1$, the bound on the delay is $b=1$, so we only need to predict the node at the root of the output of each state.
For example, the root node of $q_b(t)$ for any input tree $t$ is $b$ if $t$ contains at least one $b$, but it is $a$ if $t$ contains no $b$. 
Each of these cases corresponds to $t$ being in a regular tree language, which can be computed by a relabeling. 

In the full construction (detailed in the Appendix), we need to predict larger parts of those outputs, but for every predicted node in the output of a state $q$, the reverse image is a regular tree language, which we can compute with a relabeling. Moreover, an important first step of the full construction is to make the given MTT
\emph{nondeleting in the parameters} and \emph{nonerasing in the parameters}.
This means that every parameter that appears in the left-hand side of a rule must
also appear in the right-hand side of a rule, and, that there is no rule
with a right-hand side of the form $y_1$. These two properties guarantee that each state call
of the MTT will produce at least one output symbol; in fact, for every state $q$ and every
input tree $s$, it guarantees that $M_q(s)$ is of the form $\delta(t)$ for some $\delta\in \Delta$ and tree $t$.

We are now able to normalize $M$ into a composition of a relabeling with a one-to-one MTT:

\begin{theorem}\label{the:mtt_norm}
Given a shape preserving MTT $M$, one can effectively compute a relabeling $R$ and a one-to-one MTT $M'$ so that $R \comp M'$ is equivalent to $M$. 
\end{theorem}

Note that being a one-to-one MTT is equivalent to having $\Delta(t,u) = 0$ for all $t \in T_\Sigma$ and $u \in V(t)$. 
We show the construction on our example $M_1$. 

First, we compute the relabeling $R_1$ giving, for each input subtree, the state sequence, the delay and the root nodes of the outputs of states $q_0$, $q_m$, $q_a$ and $q_b$. 
Each input node with symbol $a$ is relabeled with symbol $(a,r)$ where $r$ gives the state sequence, the delay and the root nodes of outputs of states. 

Second, we decrease the delay of state sequences where it is positive. State sequence $q_a, q_b$ has delay $+1$. We decrease the delay through $q_b$ by replacing calls to $\<q_b,x>$ with the correct root node (given by the relabeling) and a new helper state which computes what appears below this root node. As we saw earlier, there are two possible root nodes for the output of $q_b$: $a$ and $b$.  

We replace the rule $\<q_0,a(x)> \to \<q_a,x>(\<q_b,x>)$ with the rules of the form: 
\begin{itemize}
\item $\<q_0,(a,r)(x)> \to \<q_a,x>(a(\<(q_b,1),x>))$ for each $r$ giving $a$ as root node of $q_b(t)$, 
\item $\<q_0,(a,r)(x)> \to \<q_a,x>(b(\<(q_b,1),x>))$ for each $r$ giving $b$ as root node of $q_b(t)$. 
\end{itemize}
The rules $\<q_b,a(x)> \to \<q_b,x>$ and $\<q_b,b(x)> \to b(\<q_b,x>)$ are replaced with: 
\begin{itemize}
\item $\<q_b,(a,r)(x)> \to a(\<(q_b,1),x>)$ and $\<q_b,(b,r)(x)> \to b(a(\<(q_b,1),x>))$ for $r$ giving $a$ as root node of $q_b(t)$, 
\item $\<q_b,(a,r)(x)> \to b(\<(q_b,1),x>)$ and $\<q_b,(b,r)(x)> \to b(b(\<(q_b,1),x>))$ for $r$ giving $b$ as root node of $q_b(t)$. 
\end{itemize}
The newly created helper state $(q_b,1)$ computes the subtree at path $1$ (first child of the root) of the output tree computed by $q_b$. 
To reflect this, we add rules for $(q_b,1)$ computed from the rules for $q_b$: 
\[
\begin{array}{l l}
\<q_b,(a,r)(x)> \to a(\<(q_b,1),x>)    & \text{ becomes: }~ \<(q_b,1),(a,r)(x)> \to \<(q_b,1),x> \\
\<q_b,(b,r)(x)> \to b(a(\<(q_b,1),x>)) & \text{ becomes: }~ \<(q_b,1),(b,r)(x)> \to a(\<(q_b,1),x>) \\
\<q_b,e> \to a(e)                      & \text{ becomes: }~ \<(q_b,1),e> \to e \\
\end{array}
\]
The rules above are for $r$ giving $a$ as root node of $q_b(t)$ (we need to add similar rules for $r$ giving $b$ as root node of $q_b(t)$). 
The state sequence $q_a, q_b$ has been effectively replaced with the sequence $q_a, (q_b,1)$ which has delay $0$. 
Third, we increase the delay of state sequences where it is negative. State sequence $q_m$ has delay $-1$. For each rule with $\<q_m,x>$ and an output symbol in the right-hand side, we transform the output symbol into a new helper state. 
For example, $\<q_0,b(x)> \to \<q_m,x>(b(e))$ becomes $\<q_0,(b,r)(x)> \to \<q_m,x>(b(\<(1,e),x>))$, where $(1,e)$ is a helper state computing $e$. 
We now have the state sequence $q_m, (1,e)$ with delay $0$. 
We then add rules for $(1,e)$ so that the state sequence $q_m, (1,e)$ keeps delay $0$ when $q_m$ calls itself. 
We add the rules $\<(1,e),(a,r)(x)> \to \<(1,e),x>$; $\<(1,e),(b,r)(x)> \to \<(1,e),x>$ and $\<(1,e),(e,r)> \to e$ for all label $r$. 
We have now built a one-to-one MTT which, precomposed with the relabeling $R_1$, is equivalent to $M_1$. The full construction is laid out in appendix.



\section{Conclusion and Open Problems}

It was shown that shape preservation can be decided for various
classes of tree transducers: nondeterministic top-down tree transducers (TOPs) with look-ahead,
nondeterministic bottom-up tree transducers, and compositions of total deterministic
macro tree transducers (MTT). The proofs are very short and concise, because
the rely on previously known results. 
For the nondeterministic classes shape preservation is reduced
to deciding functionality of such transducers and to deciding equivalence of
functional such transducers (both of which are known from the literature).
For MTTs (or compositions of MTTs) we first decide linear size increase and then
decide equivalence of the resulting class.

It was shown how to construct for a shape preserving TOP an equivalent
relabeling TOP. We showed how to go from non-linear to linear shape preserving TOPs
and then rely on the result of~\cite{DBLP:journals/tcs/FulopG03} to go to a relabeling TOP.
In the future it would be nice to have a self-contained procedure for the full
normalization and analyze this procedure further: can it be shown that if
the given shape preserving TOP is deterministic, then we can construct an equivalent
deterministic relabeling? Can we normalize bottom-up transducers using
this very same procedure?
For MTTs a number of open problems remain, most prominently the question whether our
results can be generalized to nondeterministic MTTs.

\bibliographystyle{splncs04}
\bibliography{bib}

\section{Appendix}

\subsubsection{Proof of Lemma~\ref{lem:state_sets}: }
If $q_i = q_j$ then there exists a tree $t \in T_\Sigma$ and a valid run whose provisional output at some path $u \in V(t)$ is: $t[u \leftarrow x] = t'[u_k \leftarrow q_k(x)]_{k \leq n}$. 
Note that the subrun for $q_i(t/u)$ (at path $u_i$ in the output) visits at least one leaf of $t/u$ (as, according to Lemma~\ref{lem:leafs}, it cannot produce a leaf from an inner node of $t/u$, and its output $t'/{u_i}$ contains at least one leaf). 
Then, if we replace the subrun for $q_j(t/u)$ (at path $u_j$ in the output) with the subrun for $q_i(t/u)$ (at path $u_i$ in the output), then we obtain a new valid run on $t$, but with at least two state calls on a same leaf in $t/u$, which is impossible according to Lemma~\ref{lem:leafs}. \qed


\subsubsection{Proof of Lemma~\ref{lem:origin_distance}:} 
Let $t[u \leftarrow x] \xrightarrow{M} t'[u' \leftarrow q(x)]$ and $S$ as above, with $t/u \xrightarrow{q} t'/{u'}$. We assume that $q \notin f(S)$. 

So there exists $(t_1,t'_1)$ and $(t_2,t'_2) \in M_q$ such that $t_1, t_2 \in I(S)$ and $\sh(t'_1) \neq \sh(t'_2)$. 
If $u$ and $u'$ were incomparable w.r.t.\ the prefix order, then we would have two runs $t[u \leftarrow t_1] \xrightarrow{M} t''_1$ and $t[u \leftarrow t_2] \xrightarrow{M} t''_2$ with $t''_1/{u'} = t'_1$ and $t''_2/{u'} = t'_2$. 
So $\sh(t''_1/{u'}) \neq \sh(t''_2/{u'})$. 
Since $M$ is shape preserving we would have: $\sh(t/{u'}) = \sh(t''_1/{u'}) = \sh(t''_2/{u'})$, a contradiction, so $u$ and $u'$ are comparable. 
So there exists a path $v$ with either $u' = u\,v$ or $u = u'\,v$. 

There exists $t_1 \in \dom(M_q)$ with height $\leq 2^{|Q|}$ (any tree of larger height must contain a loop repeating the same state set, which one could remove to get a smaller tree in $\dom(M_q)$). So there is $t'_1$ with $(t_1,t'_1) \in M_q$, and the height of $t'_1$ is $\leq \max(\rhs) . 2^{|Q|}$ where $\max(\rhs)$ is the maximum height of the right-hand side of a rule of $M$. 
We have $t[u \leftarrow x] \xrightarrow{M} t'[u' \leftarrow q(x)]$ and $t_1 \xrightarrow{q} t'_1$, so $\sh(t[u \leftarrow t_1]) = \sh(t'[u' \leftarrow t'_1])$. 

In particular, if $u = u'\,v$ then $\sh(t_1) = \sh(t'_1/v)$, so $|v| \leq \max(\rhs) . 2^{|Q|}$. 
If $u' = u\,v$ then $\sh(t_1/v) = \sh(t'_1)$, so $|v| \leq 2^{|Q|}$. \qed

\subsubsection{Proof of theorem~\ref{the:copy-removal}:}
We first compute $S_Q$ the set of valid state sets for $M$. For each $S \in S_Q$ we compute $f(S)$ the set of states $q$ such that $M_q(I(S))$ is finite (where $I(S) = \bigcap_{q \in S} \dom(q)$). For each $S \in S_Q$ and $q \in f(S)$, we compute $M_q(I(S))$. For each $t' \in M_q(I(S))$, the reverse image of $\{t'\}$ by $M_q$ is
effectively a regular tree language (e.g. by  Theorem~7.4 of~\cite{DBLP:journals/jcss/EngelfrietV85})
We compute non-deterministic top-down tree automata for all such $t' \in M_q(I(S))$. We denote by $\A$ the product automaton of all of these automata. Then $\A$ determines, on any input tree $t$, the sets $M_q(\{t\})$ for all state set $S \in S_Q$ and state $q \in f(S)$. 
We denote by $P_{\A}$ the set of states of $\A$, we denote by $T_p$ the set of trees $t$ such that there is a successful run on $t$ starting with $p$ at the root. For each $p \in P_{\A}$, we note $p(S,q) = M_q(T_p)$. 

\paragraph{Construction}
We define the new linear TOP $M' = (Q', \Sigma, \Delta, Q_0, R')$ by: 

The set of states: 
$Q' = \{ (p, v, t) \mid p \in P_{\A}, v \text{ is a path of length} \leq 2^{|Q|}, t \in T_{\Delta \cup \<Q,\{x\}>} \text{ is of height} \leq \max(\rhs).2^{|Q|} \}$ 

$Q_0 = \{ (p,\varepsilon, q_0(x)) \mid p \in P_{\A}, q_0 \text{ is an initial state of } M \}$

The set of rules $R' = R_+ \cup R_\varepsilon \cup R_-$

Rules in $R_+$:
$(p, v, q(x))(\sigma(x_1)) \to t'[u' \leftarrow (p',v',q'(x))(x_1)]$
if there is a rule $q(\sigma(x_1)) \to t'[u' \leftarrow q'(x_1)]$ in $M$ with $t' \in T_\Delta$, $v\,1 = u'\,v'$ and $p(\sigma) \xrightarrow{\A} p'$.

Rules in $R_\varepsilon$:
$(p, \varepsilon, \sigma'(t_1, \dots ,t_n))(\sigma(x_1, \dots, x_n)) \to \sigma'((p_1,\varepsilon,t'_1)(x_1), \dots, (p_n,\varepsilon, t'_n)(x_n))$
if there is a rule $p(\sigma) \xrightarrow{\A} p_1, \dots, p_n$ in $\A$ and, 
for each $i \leq n$, noting $S_i$ the set of states appearing in $t_i$:
\[t'_i = t_i[q(x_j) \leftarrow t_q]_{q \in S_i}\]
where, 
for each $q \in S_i \setminus f(S_i)$, $q(x_j)$ is replaced with $t_q$ such that there is a rule $q_i(\sigma(x_1, \dots, x_k)) \to t_q$ in $M$ and $i=j$, 
for each $q \in f(S_i)$, $q(x_j)$ is replaced with some $t_q \in p_j(S_i,q)$.

Rules in $R_-$:
$(p, v, q(x))(\sigma(x_1, \dots, x_n)) \to t'[v \leftarrow \sigma'((p_1,\varepsilon,t'_1)(x_1), \dots, (p_n,\varepsilon, t'_n)(x_n))]$
if there is a rule $p(\sigma) \xrightarrow{\A} p_1, \dots, p_n$ in $\A$ and 
a rule $q(\sigma(x_1, \dots, x_n)) \to t'[v \leftarrow \sigma'(t_1, \dots, t_n)]$ in $M$ with $t' \in T_\Delta$ and, 
for each $i \leq n$, noting $S_i$ the set of states appearing in $t_i$:
\[t'_i = t_i[q(x_j) \leftarrow t_q]_{q \in S_i}\]
where, 
for each $q \in S_i \setminus f(S_i)$: $q(x_j)$ is replaced with $t_q$ such that there is a rule $q_i(\sigma(x_1, \dots, x_k)) \to t_q$ in $M$ and $i=j$, and, 
for each $q \in f(S_i)$: $q(x_j)$ is replaced with some $t_q \in p_j(S_i,q)$. 

Note that $M'$ is linear and non-deleting, so each input node is processed by exactly one state. 

\paragraph{Correctness}
First note that a run of $M'$ on input $t$ induces a run of $\A$ on $t$ such that, if the node at path $u$ in $t$ is processed by a state $(p,v,t')$ of $M'$, then it is processed by state $p$ of $\A$. 

\begin{claim}
For all $(p,v,t_0) \in Q'$, trees $t \in T_\Sigma$ and $t' \in T_\Delta$:
\[\text{if }~~ (p,v,t_0)(t) \xrightarrow{M'} t' ~~~~\text{ then }~~ t_0[x \leftarrow t] = t' ~~\text{ and }~~ \sh(t') = \#_1^{|v|}(\sh(t))\]
\end{claim}

We can prove this claim by induction over the tree $t$. 

\begin{claim}
For all trees $t \in T_\Sigma$ and $t' \in T_\Delta$ such that there is a run $t \xrightarrow{M} t'$ whose provisional run at some path $u \in V(t)$ is of the form $t[u \leftarrow x] \xrightarrow{M} t'[u_i \leftarrow q_i(x)]$, 
there exists a tree $t_0 \in T_{\Delta \cup \<Q,\{x\}>}$ and paths $u' \in V(t')$ and $v \in V(t'/{u'})$ such that $u = u'\,v$ and, for all successful run of $\A$ on $t$ with state $p \in P_{\A}$ at path $u$: 
\[(p,v,t_0)(t/u) \xrightarrow{M'} t'/{u'}\]
\end{claim}

We can prove this claim by induction over path $u \in V(t)$. 
\qed


\subsection{Proofs for normalization of shape preserving MTTs}

\subsubsection{Finiteness of $T_\Sigma$ and $T_\Delta$}
Note that if $T_\Sigma$ is finite, then the normalization procedure is trivial: we can use a relabeling to check the whole input, then output the corresponding output one node at a time. 
If $T_\Sigma$ is infinite but not $T_\Delta$ then no total MTT can be shape preserving. 
For the rest of this part we assume that both $T_\Sigma$ and $T_\Delta$ are infinite. 

\subsubsection{Erasing input subtrees, copying and deleting parameters}
Note that a shape preserving MTT has to visit every node of its input, otherwise we could substitute the unvisited subtree with a tree of a different shape, breaking the shape preserving property. 

Any shape preserving MTT is of linear size increase (LSI), so according to~\cite{DBLP:journals/siamcomp/EngelfrietM03} it is effectively equivalent to the composition of a DBQREL with an MTT which does not erase nor copy its parameters, and has a finite set of state sequences. 

For the rest of the section we assume that $M$ is visiting every node of its input, that no parameter is ever deleted or copied, and there is a bound on the number of state sequences. 
And we have a relabeling $R$ such that $R \comp M$ is shape preserving. 

\subsubsection{Proof of Lemma~\ref{lem:delay_bound}:}
Since $M$ is deterministic and there is a bound on the number of state sequences, one can build an MTT $M'$ equivalent to $M$ where states of $M'$ are of the form $(q,S)$ where $q$ is a state from $M$ and $S$ is the state sequence at the input nodes where state $(q,S)$ is called. 

We now prove that if two inner nodes of any input trees have the same label (from the relabeling $R$) and same state sequence, then they have the same delay $\Delta$. 

Let $t, t' \in T_\Sigma$, $u \in V(t)$ and $u' \in V(t')$ such that the nodes at path $u$ in $t$ and path $u'$ in $t'$ have the same label $r$ from the relabeling $R$ and state sequence $q_1, \dots , q_n$. Then the state sequence at path $u$ in $t_0 = t[u \leftarrow t'(u')]$ is also $q_1, \dots , q_n$. 
Then $t_0[u \leftarrow X] = t[u \leftarrow X]$ and: 
\begin{align*}
	|M(t_0[u \leftarrow X])| &= |M(t[u \leftarrow X])| \\
	(\sum_{i=1}^{n} |q_i(t_0(u))| ) - |M(t_0)| &= (\sum_{i=1}^{n} |q_i(t(u))| ) - |M(t)| \\
	\Delta(t_0,u) + |t_0(u)| - |M(t_0)| &= \Delta(t,u) + |t(u)| - |M(t)| \\
	\Delta(t_0,u) + |t_0(u)| - |t_0| &= \Delta(t,u) + |t(u)| - |t| \\
	\Delta(t_0,u) - |t_0[u \leftarrow X]| &= \Delta(t,u) - |t[u \leftarrow X]| \\
	\Delta(t_0,u) &= \Delta(t,u) \\
	\Delta(t',u') &= \Delta(t,u)
\end{align*}
So $\Delta(t,u)$ only depends on the label from the $R$ and the state sequence of the node at path $u$ in $t$. 

Now we show how to compute the delay with a tree relabeling. Let $t \in T_\Sigma$ and $u \in V(t)$ with $t(u) = \sigma(t_1, \dots, t_k)$ for some input symbol $\sigma$ of arity $k$ and trees $t_1, \dots, t_k$. Let $q_1, \dots , q_n$ be the state sequence on $t$ at path $u$, and let $\rhs_1, \dots,  \rhs_n$ be the right-hand sides of the rules of $M$ for $q_1, \dots , q_n$ to the node at path $u$ in $t$. Noting $S(uj)$ the state sequence at path $uj$ in $t$ for $j \leq k$:
\begin{align*}
	\Delta(t,u) &= \sum_{i=1}^{n} |q_i(t(u))| - |t(u)| \\
	&= \sum_{i=1}^{n} |\rhs_i| + \sum_{j=1}^{k} \sum_{q \in S(uj)} |q(t_j)| - |t(u)| \\
	&= \sum_{i=1}^{n} |\rhs_i| + \sum_{j=1}^{k} \sum_{q \in S(uj)} |q(t_j)| - 1 - \sum_{j=1}^{k} |t_j| \\
	&= \sum_{i=1}^{n} |\rhs_i| - 1 + 
	\sum_{j=1}^{k} (\sum_{q \in S(uj)} |q(t_j)| - |t_j|) \\
	&= \sum_{i=1}^{n} |\rhs_i| - 1 + \sum_{j=1}^{k} \Delta(t,uj) 
\end{align*}
So $\Delta(t,u)$ can be effectively computed from the state sequence at path $u$ in $t$ (given by $R_s$), the rules of $M$ and the delays $\Delta(t,uj)$ at the child nodes of the node at path $u$ in $t$. So the delay can be computed by a relabeling, on an input already labeled with the state sequences (by the relabeling $R_s$). The composition of those relabelings is also effectively computable, which gives us a relabeling computing both the state sequences and the delays at each node. 

Conjointly, the fact that $\Delta(t,u)$ only depends on the label from $R$ and the state sequence, and the formula for computing the delay $\Delta(t,u)$ from the rules of $M$ and other delays, imply a bound $b = n.\mathsf{max}(\rhs).(k+1)^{|R_s|}$ on $|\Delta(t,u)|$, where $n$ is the maximum size of a state sequence, $\mathsf{max}(\rhs)$ is the maximum size of a right-hand size of a rule in $M$, $k$ is the maximum arity of input symbols and $|R_s|$ is the size of the relabeling $R_s$ (i.e.\ the number of labels in $R$ times the number of different state sequences). \qed

\subsubsection{Predicting outputs of state calls}
In order to produce output earlier than $M$, we change the relabeling so as to predict parts of those outputs, so we can produce parts of output earlier. 

We first define the notion of prefix tree.

\begin{definition}
	For all tree $t \in T_{\Delta}(Y)$ we say that $t' \in T_{\Delta}(X \cup Y)$ is a prefix tree of $t$ if there exists $n \in \N$ and $t_1, \dots, t_n \in T_{\Delta}(Y)$ such that $t'[X_i \leftarrow t_i]_{i \leq n}$ and for all $i \leq n$, $X_i$ occurs once in $t'$. 
\end{definition}

For each state call we will predict a finite prefix tree of its output, but containing at least as many nodes as the bound $b \in \N$ on the delay. 
For all constant $b \in \N$, state $q \in Q$ and input tree $t$, we denote by $q_{\uparrow b}(t)$ the tree $t' \in T_{\Delta}(X \cup Y)$ obtained from $q(t)$ by replacing each non-parameter subtree at depth $b+1$ in $q(t)$ with a fresh variable in $X$. Note that $q_{\uparrow b}(t)$ can also be described as the largest prefix tree of $q(t)$ of height $b$ with the most parameters $y \in Y$. 
Then if $|q_{\uparrow b}(t)| < b$ then $q_{\uparrow b}(t) = q(t)$. 

\begin{lemma}\label{lem:delay_look_ahead}
	For any constant $b \in \N$, we can construct 
	a deterministic top-downrelabeling
	on the input trees of $M$ which labels each subtree $t$ with the tuple $(q_{\uparrow b}(t))_{q \in Q}$. 
\end{lemma}

In our example $M_1$, the bound on the delay is $b=1$. 
For state $q_b$ we have: \\ 
$q_{b\uparrow 1}(t)= 
\begin{cases*}
	b(x) & if $t$ contains at least one $b$ \\
	a(x) & if $t$ contains no $b$ \\
\end{cases*}$ \\
where $x$ is a variable in $X$. 
Each of these cases corresponds to $t$ being in a regular tree language, which can be computed by a relabeling.

\subsubsection{Proof of Lemma~\ref{lem:delay_look_ahead}:}
For all input tree $t= \sigma(t_1, \dots, t_n) \in T_{\Sigma}$ and for each state $q$ of $M$, one can deduce $q_{\uparrow b}(t)$ directly from the rules of $M$ and from the values of $q'_{\uparrow b}(t_i)$ for all $i \leq n$ and all states $q'$ of $M$. This gives the rules for the bottom-up tree automaton. \qed

\subsubsection{Proof of theorem~\ref{the:mtt_norm}:}
First we use Lemma~\ref{lem:delay_bound} on $M$ to get a tree relabeling $R$ giving the state sequences and the delay at each input node, then Lemma~\ref{lem:delay_look_ahead} with the bound $b$ from Lemma~\ref{lem:delay_bound} to get a new relabeling (we take the product relabeling of this relabeling and of the previous relabeling). So for each input node we have its label from $R$, its state sequence, its delay and a part of its output (a part including all nodes of depth $\leq b$). 

Each pair $(r, S)$ where $r$ is a label from $R$ and $S$ is a state sequence, is associated with a delay $\delta \in \Z$. For $M$ to be a relabeling we need to replace state sequences that have $\delta \neq 0$ with ones that have $\delta = 0$. There are two big cases: $\delta > 0$ and $\delta < 0$. We start by removing the states sequences with $\delta > 0$.

\paragraph{Modify a rule in $M$ to remove strictly positive delays}
Let $\ell$ be a look-ahead state of $M$, $S = q_1, \dots, q_n$ a state sequence of $M$ and $\delta \in \Z$ the corresponding delay. We will modify the rules for the states $q_1, \dots, q_n$ on input nodes of the form $\sigma(t_1, \dots , t_k)$ with look-ahead $\ell$ and state sequence $S$ so that the delays $\delta_1, \dots, \delta_k$ of the subtrees $t_1, \dots, x_k$ are null or negative: $\forall i \leq k$: $\delta_i \leq 0$. 

For each subtree $t_i$ such that $\delta_i > 0$ we note $S_i = q'_1, \dots, q'_m$ the state sequence at node $t_i$. If $\sum_{j=1}^{m}|q'_{j \uparrow b}(t_i)| < b$ then $\sum_{j=1}^{m}|q'_{j \uparrow b}(t_i)| = \sum_{j=1}^{m}|q'_{j}(t_i)| = \delta_i + |t_i| \geq \delta_i$, otherwise $\sum_{j=1}^{m}|q'_{j \uparrow b}(t_i)| \geq b \geq \delta_i$ by Lemma~\ref{lem:delay_bound}. In either case $\sum_{j=1}^{m}|q'_{j \uparrow b}(t_i)| \geq \delta_i$, let $t'_1, \dots, t'_m$ be prefixes of $q'_{1 \uparrow b}(t_i), \dots, q'_{m \uparrow b}(t_i)$ respectively, such that $\sum_{j=1}^{m}|t'_j| = \delta_i$. For each $j \leq m$ and $X_c$ in $t'_j$ we note $p_c$ the path to $X_c$ in $t'_j$ and we create a helper state $\<q'_j,p_c>$ which, on input $t_i$ will compute the subtree of $q'_j(t_i)$ at path $p_c$. 


In the right-hand sides of the rules for $q_1, \dots, q_n$ on $\sigma(t_1, \dots , t_k)$ (as mentioned above), we replace each call $q'_j(x_i)$ with $t'_j[X_c \leftarrow \<q'_j,p_c>(x_i)]$ where $\<q'_j,p_c>$ is a helper state. The state sequence $S_i = q'_1, \dots, q'_m$ on subtree $t_i$ is then replaced with $S'_i = \<q'_1,p_{1,1}>, \dots, \<q'_m,p_{m,d_m}>$ and the delay at $t_i$ is then $\sum_{j=1}^{k} \sum_{c=1}^{d_j} |\<q'_j,p_{j,d_j}>(t_i)| - |t_i| = \sum_{j=1}^{m}|q'_{j}(t_i)| - \sum_{j=1}^{m}|t'_j| - |t_i| = \delta_i - \sum_{j=1}^{m}|t'_j| = 0$ (where $d_j$ is the number of variables $X_c$ in $t'_j$). 

We apply this construction to all the rules in $M$, so all state calls appearing on the right side of rules have a not strictly positive delay. Note that at the root of an input tree $t \in T_\Sigma$, the state sequence is $q_0$ and the delay is also $\Delta(t,\varepsilon) = |q_0(t)| - |t| = 0$.

\begin{claim}
	If helper states are such that $\<q'_j,p_c>(t_i) = (q'_j(t_i))(p_c)$, then the new rules effectively compute the same output as with $M$. 
\end{claim}

\paragraph{Create new rules for helper states}
We now create rules for the helper states. For each helper state $\<q,p>$ and rule of the form $q(\sigma(x_1, \dots, x_k), y_1, \dots, y_m) \to \rhs$, we add the rule 
$\<q,p>(\sigma(x_1, \dots, x_k), y_1, \dots, y_m) \to (\rhs)_{\downarrow p}$ where the $(~)_{\downarrow p}$ operation is defined by induction as:
\begin{itemize}
	\item $(t)_{\downarrow \varepsilon} = t$
	\item $(\sigma(t_1, \dots, t_n))_{\downarrow ip} = t_i$ for $i \leq n$
	\item $(q'(x_i,t_1, \dots, t_m))_{\downarrow p_1 p_2} = (t_j)_{\downarrow p_2}$ for $q' \in Q$ if $q'_{\uparrow b}(x_i)/{p_1}=y_j$ with $j \leq m$
	\item $(q'(x_i,t_1, \dots, t_m))_{\downarrow p} = \<q',p>(x_i,t_1, \dots, t_n)$ for $q' \in Q$ if $q'_{\uparrow b}(x_i)/p=t'$ with $t' \notin Y$
\end{itemize}
Note that other cases are impossible (from the way the right-hand side of rules and the helper states were built). 

We can show by induction that these rules ensure that $\<q,p>(t) = (q(t))/p$ for all helper state $\<q,p>$ and input tree $t$. 

\paragraph{Modify a rule in $M$ to remove negative delays}
Let $\ell$ be a look-ahead state of $M$ and $S = q_1, \dots, q_n$ a state sequence of $M$ such that the pair $(\ell,S)$ is accessible in $M$ and the corresponding delay is $\delta = 0$. We will modify the rules for the states $q_1, \dots, q_n$ on input nodes of the form $t=\sigma(t_1, \dots , t_k)$ with look-ahead $\ell$ and state sequence $S$ so that the delays $\delta_1, \dots, \delta_k$ of the subtrees $t_1, \dots, t_k$ become $0$. Note that such a rule must exist since the initial state has delay $0$. 

Note that since $\delta = 0$ we have: $0 = \sum_{i=1}^{n} |q_i(t)| - |t| = \sum_{i=1}^{n} |\rhs_i| + \sum_{j=1}^{k} (\delta_j + |t_j|) - |t| = \sum_{i=1}^{n} |\rhs_i| + \sum_{j=1}^{k} \delta_j - 1$, where $\rhs_i$ is the right-hand side of the rule for $q_i$ applied at the root of $t$. So $\sum_{i=1}^{n} |\rhs_i| = 1 - \sum_{j=1}^{k} \delta_j$. 
We add helper states of the form $\< a, c>$ where $a$ is an output tree constant of arity $g$, $c \in \N$ is such that $c \leq b$ and the helper state $\< a, c>$ takes $g$ parameters. Such helper states are used to delay production of output, for each such state of the form $\<a,c>$ we add the rule $\<a,c>(x,y_1, \dots, y_g) \to a(y_1, \dots, y_g)$. 
The constant $c$ is used to avoid having several times the same state in a state sequence, this allows to produce one output tree node $a$ without producing all of them at the same time. Note that for all helper state $\<a,c>$ and input tree $t$ we have $|\<a,c>(t)| = |a(y_1, \dots, y_g)| = 1$. 

For each subtree $t_j$ such that $\delta_j < 0$ we select $|\delta_j|$ nodes labeled with an output tree constant from the right-hand sides of the rules for $q_i$ at the root of $t$, and we replace each such node labeled $a$ with a state call $\<a,c>(x_j)$ (note that this is possible because the arity of state call $\<a,c>(x_j)$ is the same as that of nodes labeled $a$). If two such nodes have the same label $a$ then we use different helper states $\<a,c_1>$ and $\<a,c_2>$ with $c_1 \neq c_2$ (this is possible because $|\delta_j| \leq b$ and there are $b$ different helper states of the form $\<a,c>$). 

Noting $S_j = q'_1, \dots, q'_m$ the former state sequence at node $t_j$, with the added helper states we now have the state sequence $S'_j = q'_1, \dots, q'_m, h_1, \dots, h_{\delta_j}$ where $h_1, \dots, h_{\delta_j}$ are new helper states. The delay of this state sequence is then $\delta_j + |\delta_j| = 0$. 

\paragraph{Iterate the construction to remove all negative delays}
The construction described above only works when the delay at the current node is $0$, it adjusts the delay of the child nodes to $0$. This construction needs to be iterated as new look-ahead/state sequence pairs $(\ell,S)$ with delay $0$ become accessible. This iteration terminates when all accessible pairs $(\ell,S)$ have delay $0$, it has to terminate because there is a bound $b$ on the number of newly added helper states in any state sequence. 

When the iteration terminates, all accessible pairs have delay $0$. 
\qed


\end{document}